\documentclass[3p,10pt,a4paper,fleqn]{yujor1}

\newtheorem{thm}{Theorem}[section]

\newtheorem{cor}{Corollary}[section]

\newtheorem{dfn}{Definition}[section]
\newtheorem{rem}{Remark}[section]

\usepackage{amssymb}
\usepackage{amsmath}
\usepackage{textcomp}
\usepackage{graphicx}
\usepackage{url}
\usepackage{listings}
\usepackage{bm}
\usepackage{booktabs}
\usepackage{multirow}
\usepackage{geometry}
\usepackage{xcolor}
\usepackage[norelsize]{algorithm2e}
\usepackage{changepage}
\usepackage{accents}
\usepackage{float}
\usepackage[inline]{enumitem}
\usepackage{hyperref}

\newlength{\dhatheight}
\newcommand{\dhat}[1]{%
    \settoheight{\dhatheight}{\ensuremath{\hat{#1}}}%
    \addtolength{\dhatheight}{-0.35ex}%
    \hat{\vphantom{\rule{1pt}{\dhatheight}}%
    \smash{\hat{#1}}}}

\newcommand{\YujorVol}{28}
\newcommand{\YujorNo}{4}
\newcommand{\YujorYear}{2018}
\newcommand{\YujorPages}{475--499}
\newcommand{\YujorDOI}{10.2298/YJOR180520022B}
\newcommand{\YujorRCVDate}{May 2018.}
\newcommand{\YujorACPDate}{August 2018.}

\begin{document}



\title {APPROXIMATE CALCULATION OF TUKEY'S DEPTH AND MEDIAN WITH HIGH-DIMENSIONAL DATA
}

\YujorAuthor{Milica BOGI\'CEVI\'C}{School of Electrical Engineering, University of Belgrade, Belgrade, Serbia}{antomripmuk@yahoo.com}

\YujorAuthor{Milan MERKLE}{School of Electrical Engineering, University of Belgrade, Belgrade, Serbia}{emerkle@etf.rs}

\newcommand{\YujorAuthorNames}{M. Bogi\'cevi\'c, M. Merkle}

\newcommand{\YujorShortTitle}{Calculation of Tukey's depth
}

\newcommand{\YujorMSC}{62G15, 68U05}
\newcommand{\YujorKeywords}{Big data,  multivariate medians, depth functions, computing Tukey's depth.}

\begin{abstract}
We present a new fast approximate algorithm for Tukey (halfspace) depth level sets and its implementation-ABCDepth.
Given a $d$-dimensional data set
for any $d\geq 1$, the algorithm is based on a representation of level sets as intersections of balls
in $\mathbb{R}^d$. Our approach does not need calculations of projections of sample points to directions.
This novel idea enables calculations of approximate level sets in very high dimensions with complexity which is linear in $d$,
which provides a great advantage over all other approximate algorithms. Using different versions of this algorithm we demonstrate approximate
calculations of the deepest set of points ("Tukey median") and  Tukey's depth of a sample point or  out-of-sample point, all with a
linear in $d$ complexity. An additional theoretical advantage of this approach is that the data points are not assumed to be in "general position".
Examples with real and synthetic data show that the executing time of the algorithm in all mentioned
versions in high dimensions is much smaller  than the time of other implemented algorithms. Also, our algorithms can be used with
thousands of multidimensional observations.

\end{abstract}

\maketitle




\newpage

\section{INTRODUCTION}
\label{seci}

\medskip

 Although this paper is about multivariate medians and related notions, for completeness and understanding rationale of multivariate setup,
 we start from the univariate case. In terms of probability distributions, let $X$ be a random variable and let $\mu=\mu_X$ be the corresponding distribution, i.e., a probability
measure
on $(\mathbb{R},\mathcal B)$ so that $P(X\leq x)= \mu \{ (-\infty, x ]\}$.  For univariate case, a median of $X$  (or a median of $\mu_X$) is
any number $m$ such that
$P(X\leq m)\geq 1/2$ and $P(X\geq m) \geq 1/2$.  In terms of data sets, this property means that to reach any median point from the outside of the data set,
we have to pass at least $1/2$ of data points, so this is the deepest point within the data set. With respect to this definition, we can define the depth of
any point $x\in \mathbb{R}$ as
\begin{equation}
\label{depd1}
D(x,\mu)=\min\{P(X\leq x), P(X\geq x)\} =\min \{\mu ((-\infty, x]), \mu ([x, +\infty))\}.
\end{equation}
The set of all median points $\{\rm Med \mu\}$ is  a non-empty compact interval (can be a singleton).
As shown in \cite{merkle05},
\begin{equation} \label{unimed}
\{\rm Med \mu\} =  \bigcap_{J=[a,b]:\ \mu (J) >1/2} J,
\end{equation}
 and (\ref{unimed}) can be taken for an alternative (equivalent) definition of univariate median set.  In $\mathbb{R}^d$ with $d>1$, there are quite a few
 different concepts of depth and medians (see  for example \cite{survey15}, \cite{small90}, \cite{zuoserf00}). In this paper we
 propose an algorithm for  halfspace depth (Tukey's depth, \cite{tukey75}), which is based on extension and generalization of (\ref{unimed}) to $\mathbb{R}^d$
 with balls in place of  intervals as in  \cite{merkle10}.

 The rest of the paper is organized as follows. Section 2 deals with a   theoretical background of the algorithm in a broad sense.
 In  Section 3 we present an approximate algorithm for finding Tukey median as well as  versions of the same algorithm for finding Tukey depth of a sample point,
the depth of  out-of-sample point,  and for  data contours. We also provide a derivation of complexity for each version of the algorithm and
present examples. Section 4 provides a comparison with several other algorithms in terms of  performances.

\medskip


\section {THEORETICAL BACKGROUND: DEPTH FUNCTIONS BASED ON FAMILIES OF CONVEX SETS}
\label{sect}

\medskip

\begin{dfn} \label{ddepthf} Let $\mathcal V$ be a family of convex sets in $\mathbb{R}^d$, $d\geq 1$, such that:
(i) $\mathcal V$ is closed under translations and
(ii) for every ball $B\in \mathbb{R}^d$ there exists a set $V\in \mathcal V$ such that $B\subset V$.
Let $\mathcal U$ be the collection of complements of sets in $\mathcal V$.  For a given probability measure $\mu$ on $\mathbb{R}^d$, let us  define
\begin{equation}
 \label{depdg}
D(x;\mu, \mathcal V)= \inf \{\mu(U)\; |\; x\in U\in \mathcal U\}= 1-\sup\{\mu(V)\; |\; V\in \mathcal V,\ x\in V'\}
\end{equation}
The function $x\mapsto D(x;\mu,\mathcal V)$ will be called a depth function based on the family $\mathcal V$.
\end{dfn}

\begin{rem}{\rm Definition \ref{ddepthf} is a special case of Type $D$ depth functions as defined in \cite{zuoserf00} which can be obtained by
generalizations of (\ref{depd1}) to higher dimensions.
The conditions stated in \cite{merkle10} that provide a desirable behavior of the depth function are satisfied in this special case,
 with additional requirements
that sets in $\mathcal V$ are closed or compact. For instance, taking $\mathcal V$ to be the family of all "boxes" with sides parallel
to coordinate hyper-planes
yields the deepest points which coincide with the coordinate-wise median. It is easy to see that in this case,
 regardless of the measure $\mu$, there exists at least one point $x\in \mathbb{R}^d$ with $D(x;\mu,\mathcal V)\geq \frac{1}{2}$. According to
the next theorem \cite[Theorem 4.1]{merkle10}, in general case of arbitrary $\mathcal V$, the maximal depth can't be smaller than $1/(d+1)$.

  }
\end{rem}

\begin{thm} \label{alpha} Let $\mathcal V$ be any non-empty family of compact convex subsets of $\mathbb{R}^d$ satisfying the conditions as in
Definition \ref{ddepthf}.
Then for any
probability measure $\mu$ on $\mathbb{R}^d$ there exists a point $x\in \mathbb{R}^d$ such that $D(x;\mu,{\mathcal V})\geq \frac{1}{d+1}$.
\end{thm}

The set of points with maximal depth is called {\em the center of distribution}, denoted  as $ C(\mu,\mathcal V)$.
In general, one can observe {\em level sets} (or {\em depth regions} or {\em depth-trimmed regions}) of level  $\alpha$ defined as
\begin{equation}
 \label{salpha}
 S_{\alpha}= S_{\alpha}(\mu,\mathcal V):=\{ x\in \mathbb{R}^d\; | \; D(x; \mu,\mathcal V)\geq \alpha \} .
\end{equation}

Clearly, if $\alpha_1<\alpha_2$, then $S_{\alpha_1} \supseteq S_{\alpha_2}$ and $S_{\alpha}=\emptyset$ for $\alpha>\alpha_{m}$, where $\alpha_m$ is
the maximal depth for a given probability measure $\mu$.

The borders of depth level sets  are called {\em depth contours} (in two dimensions)  or {\em depth surfaces} in general. Depth surfaces  are of
 interest in multivariate statistical inference (see \cite{doga92,dughocha11,zhouserf08,rouru99}). Although (\ref{salpha}) suggests
that in order to describe level sets we   need to calculate depth functions, there is another way, as the next result shows
(\cite{dono82}, \cite{zuoserf00} and Theorem 2.2. in \cite{merkle10}).

\begin{thm} \label{repls}
Let $D(x;\mu,\mathcal V)$ be defined for $x\in \mathbb{R}^d$ as in Definition \ref{ddepthf}. Then for any $\alpha\in (0,1]$
\begin{equation}
\label{lsvi}
S_{\alpha}(\mu,\mathcal V) =
\bigcap_{V\in \mathcal V, \mu(V)>1-\alpha} V.
\end{equation}
\end{thm}

From (\ref{lsvi}) it follows that the center of a distribution is  the smallest non-empty level set, or equivalently,
\begin{equation}
\label{cenf}
 C(\mu,\mathcal V)=\bigcap_{\alpha:S_\alpha \neq \emptyset} S_{\alpha} (\mu,\mathcal V)
\end{equation}

Since sets in $\mathcal V$ are convex, the level sets are also convex. From (\ref{salpha}) and (\ref{lsvi}) we can see that
the depth function can be uniquely reconstructed starting
from level sets as follows.

\begin{cor} \label{deviale}For given $\mu$ and $\mathcal V$, let $S_{\alpha}$, $\alpha\geq 0$ be defined as in (\ref{lsvi}),
with $S_{\alpha} =\emptyset$ for $\alpha >1$.  Then the function $D: \mathbb{R}^d \mapsto [0,1]$
defined by
\begin{equation}
\label{deplf}
D(x) = h\quad \iff \quad x\in S_{\alpha}, \alpha \leq h \quad {\rm and}\quad x\not\in S_{\alpha}, \alpha >h
\end{equation}
is the unique depth function  such that (\ref{salpha}) holds.
\end{cor}

The algorithm that we propose  primary finds  approximations to level sets based on formula (\ref{lsvi}),
and then finds (approximate) depth via Corollary \ref{deviale}. The algorithm will be demonstrated in the case of half-space depth,
which is described in the next section.

The scientific interest in algorithms for Tukey's depth is shared between Statistics (robust estimation of location parameters, clustering,
classification, outlier detection, especially  with big data in high dimensions) and Computational geometry (as a challenge).

\medskip

\section{ABCDEPTH ALGORITHM FOR TUKEY DEPTH: IMPLEMENTATION AND THE OUTPUT}
\label{seca}

The most prominent representative of type $D$ depth functions is Tukey's  (or halfspace) depth, which is
defined by (\ref{depdg}) with $\mathcal U$ being a family of all open half-spaces, and $\mathcal V$ a family of closed half-spaces, as in the following
definition.
\begin{equation}
\label{odtd}
 D(x;\mu)=  \inf \{\mu(H)\; |\; x\in H\in \mathcal H\},
\end{equation}

where $\mathcal H$ is the family of all open half-spaces. In this section, we consider only half-space depth, so we use the notation $D(x,\mu)$ instead of
$D (x; \mu,{\mathcal V})$.

The idea traces back to J. W. Tukey's lecture notes \cite{tukeymim} and the conference paper  \cite{tukey75}.
It was first formalized in D. L. Donoho's Ph. D. qualifying paper  \cite{dono82} of 1982,
  a technical report \cite{gado87}  and the  1992's paper  \cite{doga92}.

One depth function can be defined based on different families $\mathcal V$.
We say that families $\mathcal V_1$ and $\mathcal V_2$ are
depth-equivalent if $D(x;\mu, {\mathcal V_1})=D(x;\mu, {\mathcal V_2})$ for all $x\in \mathbb{R}^d$ and all probability measures $\mu$.
Sufficient conditions for depth-equivalence are given in  \cite[Theorem 2.1]{merkle10}, where it was shown  that in the case of
half-space depth the following families are depth-equivalent:
a) Family of all open half-spaces; b) all closed  half-spaces; c) all convex sets; d) all compact convex sets; e) all closed or open balls.

For determining level sets we choose closed balls,  and so we can define $\mathcal V$ as a set of all closed balls (hyper-spheres)
and  level sets can be described as
\begin{equation}
\label{sabs}
S_{\alpha}(\mu,\mathcal V) =
\bigcap_{B\in \mathcal V, \mu(B)>1-\alpha} B,
\end{equation}
instead of using the classical approach based on half-spaces. The advantage of the formula (\ref{sabs}) over the intersections of half-spaces is
obvious if we recall that a ball $B$  is already the intersection of all tangent spaces that contain $B$.

From now on, we consider only the case where the underlying probability measure $\mu$ is derived from a given data set.

\subsection{The sample version}\label{saver}
In the data setup with a sample $x_1, \ldots, x_n$ (allowing repetitions) and  the counting measure
\begin{equation}
\label{cmes}
 \mu (A)= \frac{ \# \{x_i: \ x_i \in A\}}{n},
\end{equation}
it is a common practice to express the depth as an integer defined as
\begin{equation}
\label{odtds}
 D(x)=  \min \{\#\{x_i:\ x_i \in H \} \; |\; x\in H\in \mathcal H\},\qquad x\in \mathbb{R}^d,
\end{equation}
whereas the depth in terms of  probability is $D(x)/n$. As it was first noticed  by Donoho \cite{dono82}, the depth can be expressed via
one-dimensional projections to  directions  determined by normal vectors of hyperplanes  that are borders of
 half-spaces (with $d=2$ the borders are straight lines).

\begin{equation}
\label{projpt}
 D(x)= \min_{\|u\|=1} \#\{x_i: \langle u, x_i \rangle \leq \langle u,x \rangle \},
 \end{equation}
 where $\langle \cdot, \cdot \rangle $ is the inner product of vectors in $\mathbb{R}^d$.

The level set in the sample version with $\alpha=k/n$ is defined by

\begin{equation}
 \label{salphas}
 S_{\frac{k}{n}} =\{ x\in \mathbb{R}^d\; | \; D(x)\geq \lfloor n(1 - \alpha) + 1 \rfloor \} = \{ x\in \mathbb{R}^d\; | \; D(x)\geq n-k+1 \},
\end{equation}

and (\ref{sabs}) becomes

\begin{equation}
\label{sabsds}
S_{\frac{k}{n}} = \bigcap_{B\in \mathcal V, \#\{x_i:\ x_i \in B \} \geq  n-k+ 1} B,
\end{equation}
where $\mathcal V$ is a family of all closed balls in $\mathbb{R}^d$.

All so far implemented algorithms (both exact and approximate) are based on calculation of the depth  from the
formula (\ref{projpt}) or its variations.  The approximate algorithm that we propose uses a completely new approach: we start with
a discrete approximation to level sets using formula (\ref{sabs}), and then we calculate the depth in conjunction
with the formula (\ref{deplf}). In the next  subsections we present details of the approximations.

\subsection{First approximation: finite intersection.}\label{firapp}
 For a fixed sample size $n$ and a fixed $k$, for simplicity   we write $S$ instead of  $S_{\frac{k}{n}}$.
 So, for fixed $n$ and $k$, $S$ is an exact (unknown)
level set as in (\ref{salphas}) and (\ref{sabsds}). Let us choose $M$ points to be centers of balls. In most of examples in this paper,
we  choose the  points  from the sample ($M=n$), and then add  new points at random if needed.
The radius of each ball $B_i$ with the center at $c_i$  is  equal to the $(n-k +1)$th smallest
distance between $c_i$ and the points in the sample set $\{ x_1,\ldots, x_n\}$.
 In this way, we end up with the first approximation of the level set $S=S_{\frac{k}{n}}$:
\begin{equation}
\label{fapp}
\hat{S}_M := \bigcap_{i=1}^M {B_i}
\end{equation}
It is natural to assume that if we want to intersect more than $M$ balls, then we just add new balls to the existing intersection, hence
the sets $\hat{S}_M$ are nested and
\begin{equation}
\label{nesN}
\hat{S}_M \supseteq \hat{S}_{M+1} \supset S, \qquad M=1,2,\ldots
\end{equation}

As shown in \cite[Lemma 2]{Akss2010}, to decide whether or not the intersection in (\ref{fapp}) is empty, it takes $M\cdot 2^{O(d^2)}$ time,
so the exact approach is not feasible. On the other hand, for a given point $s\in  \mathbb{R}^d $ it is easy to decide whether or not it belongs to
$\hat S$ as defined in (\ref{fapp}). This observation leads to the second approximation as follows.

\subsection{Second approximation: a discrete set of points.}\label{secdis} Let $A_n=\{x_1,\ldots, x_n\}$ (set of sample points), and
let $D\supset A_n$ be a convex domain in $\mathbb{R}^d$. For $N=n+1,n+2,\ldots $, let $A_N$ be a set obtained from $A_{N-1}$  by
randomly adding one  point from $D$. Then we have

\begin{equation}
\label{nestA}
A_{N+1}\supset A_{N},\quad N=n,n+1,n+2,\ldots,
\end{equation}
and  every set $A_N$  contains all sample points. We call sets  $A_N$ ($N\geq n+1$) {\em augmented data sets}. Points
in $A_N\setminus A_n$ will be called {\em  artificial points} and their role is to "shed light" on a depth region via discrete approximation
as follows:

\begin{equation} \label{setlevp}
\dhat{S}_{M,N} :=A_N\cap \hat{S}_{M}= \{ a\in A_N \; | \; a\in \hat{S}_{M},\}
\end{equation}
where $\hat{S}_{M}$ is defined in (\ref{fapp}).
Let us note that $A_N$ does not depend on $k$, so the same $A_N$
can be used in approximation of all depth regions.

From (\ref{nestA}) and (\ref{setlevp}), it follows that

\begin{equation}
\label{nestS}
\dhat{S}_{M,N} \subseteq \dhat{S}_{M,N+1}\subset \hat{S}_{M}
\end{equation}

Therefore, for a fixed $M$, the sets $\dhat{S}_{M,N}$ approximate $\hat{S}_{M}$ from inside, obviously with  accuracy which  increases with $N$.
In order to make a  contour, we can construct a convex hull of $\dhat{S}_{\alpha}$
using \textit{QuickHull} algorithm, for example. The relations in (\ref{nestS}) remain true with  ${\rm Conv\; }(\dhat{S})$ in place of $\dhat{S}$.
Due to convexity of  $\hat{S}_M$, the approximation  with ${\rm Conv\; }(\dhat{S})$ would be better,
but then the complexity would be too high for really high dimensions. As an alternative, we use
$\dhat{S}_{M,N}$ as the final approximation to the true level set $S=S_{\frac{k}{n}}$.

\subsection{More about artificial points} \label{sectartd}   The simplest way to implement the  procedure  from \ref{firapp}
and \ref{secdis} is to  take $M=N=n$,  and  to use sample points for centers of balls and also for
 finding $\dhat{S}$ in (\ref{setlevp}).
The basic Algorithm 1 of the subsection \ref{subsectm} is presented in that setup.
However, in some cases this approach does not work, regardless of the sample size.
As an example, consider a uniform distribution in the region bounded by circles  $x^2+y^2=r^2_i$,
$r_1=1$ and $r_2=2$. It is easy to prove (see also \cite{dughocha11}) that the depth monotonically increases from $0$ outside of the
larger circle, to $1/2$ at the origin, which is the true and unique median.
With a sample from this distribution, we will not have data points inside the inner circle, and we can not identify the median in the way proposed above.

In similar cases and whenever we have sparse data or small sample size $n$, we can still visually identify depth regions and center
simply by adding artificial points to the data set.
Let the data set contain points $x_1, \ldots, x_n$ and let $x_{n+1},\ldots, x_{N}$ be points chosen from uniform distribution in
some convex domain that contains the whole data set. Then we use augmented data set (all $N$ points) in (\ref{setlevp}), but
 $n$ in formulas (\ref{cmes}) and (\ref{sabsds}) remains to be  the cardinality of the original data set.

Figure \ref{ring} shows the output of ABCDepth algorithm in the ring example above. By adding artificial data points, we are able to
obtain an approximate position of the Tukey's  median.

As another example, let us consider a triangle  with vertices $A(0, 1)$, $B(-1, 0)$ and $C(1, 0)$.
Assuming that $A,B,C$ are sample points, all points in the interior and on the border of $ABC$ triangle have depth $1/3$,
so the depth reaches its maximum value at $1/3$.
Since the original data set contains only $3$ points, by adding artificial data and applying ABCDepth algorithm we can
visualize  the Tukey's median set as shown in Figure \ref{triangle}.

\begin{figure}[htb]
    \centering
    \includegraphics[width=0.5\textwidth, height=0.4\textheight]{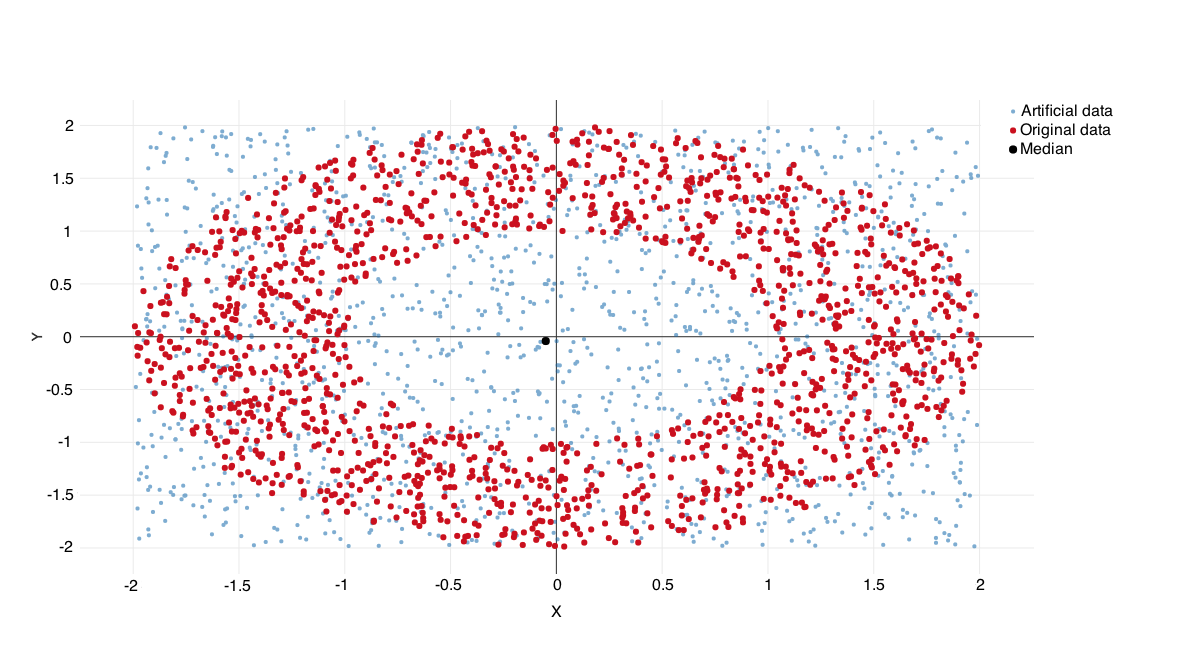}
    \caption{A sample from uniform distribution in a ring (red): Tukey's median (black) found with the aid of artificial points (blue).  }
    \label{ring}
\end{figure}

\begin{figure}[H]
    \centering
    \includegraphics[width=0.5\textwidth, height=0.4\textheight]{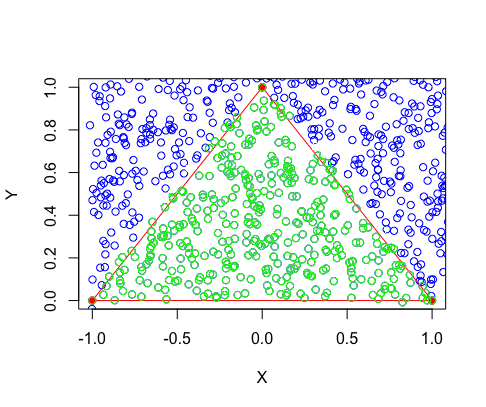}
    \caption{Tukey's median set of red triangle represented as triangle itself and green points inside of the triangle. Blue points
   are artificial data. }
    \label{triangle}
\end{figure}

In the rest of this section, we  describe the details of implementation of the approximate algorithm for finding Tukey's median,
as well as  versions of the same algorithm for finding Tukey depth of a sample point,
the depth of  out-of-sample point,  and for  data contours. For simplicity, in the rest of the paper we
use notation $S_{\alpha}$ instead of $\dhat{S}_{\alpha}$
unless explicitly noted otherwise.

\subsection{Implementation: Algorithm 1 for finding deepest points (Tukey's median)}
\label{subsectm}
\hfill\\

\textbf{Phase 1} In order to execute the calculation in (\ref{sabsds}), constructing balls for intersection is the first step. Each ball is defined by $\lfloor n(1 - \alpha) + 1 \rfloor$ nearest points to its center, so at the beginning of this phase, we calculate
Euclidean inter-distances.
Distances are stored as a triangular matrix in a \textit{list of lists} structure, where $i$-th list $(i = 1, . . . , n-1)$ contains distances $d_{i+1,j} , j = 1,\ldots  , i$. This part of the implementation is described in lines $1-6$ of Algorithm 1.

\textbf{Phase 2} In this phase, ABCDepth sorts distances for each point (center) and populates a \textit{hashmap} structure, where the \textit{key} is a center of a ball, and \textit{value} is a \textit{list} with
$\lfloor n(1 - \alpha) + 1 \rfloor$ nearest points. This part of algorithm is presented in lines $7-10$ of Algorithm 1.

\textbf{Phase 3}  Now, we intersect balls iteratively and in each iteration, $\alpha$ is increased by $\frac{1}{n}$. Since this algorithm is meant to find the deepest location, there is no need to start with the minimal value of $\alpha = \frac{1}{n}$; due to Theorem \ref{alpha}, we set the initial value of  $\alpha$ to be
 $\frac{1}{d+1}$.  Balls intersections are shown on Algorithm 1, lines $11-18$.

If the input set is sparse, ABCDepth optionally creates an augmented data set of total size $N$ as explained on
page \pageref{secdis} and demonstrated on Figures \ref{ring} and  \ref{triangle}. The rest of the algorithm takes three phases which we described above.

The initial version of ABCDepth algorithm was presented in  \cite{bome15}.

\begin{algorithm}[H]
  \LinesNumbered
  \SetAlgoLined
  \KwData{Original data, $X_n = (\boldsymbol{x_1}, \boldsymbol{x_1},...,\boldsymbol{x_n}) \in \mathbb{R}^{d \times n}$}
  \KwResult{List of level sets, $S = \{S_{\alpha_1}, S_{\alpha_2},...,S_{\alpha_m}\}$, where $S_{\alpha_m}$ represents a Tukey median}
  \tcc{Note: $S_{\alpha}$ here means $\dhat{S}_{\alpha}$}
  \BlankLine
  \For{$i \gets 2$ \textbf{to} $n$} {
	\For{$j \gets 1$ \textbf{to} $i-1$} {
		 	Calculate Euclidean distance between point $\boldsymbol{x_i}$ and point $\boldsymbol{x_j}$ \;
			Add distance to the list of lists \;
		
	}
  }
   \BlankLine
   \For{$i \gets 1$ \textbf{to} $n$} {
   	Sort distances for point $\boldsymbol{x_i}$  \;
	Populate structure with balls \;
   }
   \BlankLine
   \tcc{Iteration Phase}
   $size = n$, $\alpha_1=\frac{1}{d+1}$, $k=1$ \;
   \While{$size > 1$} {
   	$S_{\alpha_k} = \{ \bigcap_{j}^{n} B_j,  \left | B_j \right | = \lfloor n(1 - \alpha_k) + 1 \rfloor \}$ \;
	$size = \left | S_{\alpha_k} \right |$ \;
	$\alpha_{k+1}=\alpha_k + \frac{1}{n}$ \;
	Add $S_{\alpha_k}$ to $S$ \;
	$k = k + 1$ \;
   }
  \caption{Calculating Tukey median.}
\end{algorithm}

\subsubsection{Complexity}
\begin{thm}
\label{complexitytm}
ABCDepth algorithm for finding approximate Tukey median has order of $O((d + k)n^2 + n^2\log{n})$ time complexity,
where $k$ is the number of iterations in the iteration phase.
\end{thm}
\begin{proof}
To prove this theorem we use the pseudocode of  Algorithm 1.
Lines 1-6 calculate Euclidean inter-distances of points. The first \textit{for} loop (line 1) takes all $n$ points, so its complexity is $O(n)$.
Since there is no need to calculate $d(x_i, x_i)$ or $d(x_j, x_i)$ if it is already calculated, the second for-loop (line 2) runs in
 $O(\frac{n-1}{2})$ time. Finally, calculation of Euclidean distance takes $O(d)$ time. The overall complexity for lines 1-6 is:
\begin{equation}
\label{one}
O(\frac{nd(n-1)}{2}) \sim O(dn^2)
\end{equation}

Iterating through the \textit{list of lists} obtained in lines 1-6, the first \textit{for} loop (line 7) runs in $O(n)$ time.
For sorting the distances per each point, we use \textit{quicksort} algorithm that takes $O(n \log n)$ comparisons to sort $n$ points
\cite{hoare61}. Structure populating takes $O(1)$ time. Hence, this part of the algorithm has complexity of:
\begin{equation}
\label{two}
O(n^2 \log n).
\end{equation}

In the last phase (lines 11-18), algorithm calculates level sets by intersecting balls constructed in the previous steps.
In every iteration (line 12), all $n$ balls that contain $\lfloor n(1 - \alpha_k) + 1 \rfloor$ are intersected (line 13).
The parameter $k$ can be considered as a number of iterations, i.e. it counts how many times the algorithm enters in \textit{while} loop.
Each intersection has the complexity of $O(\lfloor n(1 - \alpha_k) + 1 \rfloor) \sim O(n)$
 due to the property of the hash-based data structure we use (see for example \cite{fastset}). We can conclude that the
 iteration phase has complexity of:
\begin{equation}
\label{three}
O(kn^2).
\end{equation}

From (\ref{one}), (\ref{two}), and (\ref{three}),
\begin{equation}
\label{alg1}
O(dn^2) + O(n^2 \log n) + O(kn^2) \sim O((d + k)n^2 + n^2 \log n),
\end{equation}
which ends the proof. \end{proof}

\begin{rem}{\rm \label{kmaxgp} $1^{\circ}$  From the relations between $S_{\alpha }$, $\hat{S}_{\alpha}$,
and $\dhat{S}_{\alpha}$ (in notations as in \ref{saver}, page \pageref{saver}), it follows that the maximal approximative depth for a given point
can not be greater  than its exact  depth.

$2^{\circ}$ Under the assumption that data points are in the general position, the exact sample maximal depth is $\alpha_m=\frac{m}{n}$, where $m$ is
not greater than $\lceil\frac{n}{2}\rceil$ (see \cite[Proposition 2.3]{doga92}), and so by remark $1^{\circ}$,
the number $k$ of steps satisfies the inequality

\begin{equation}
\frac{k-1}{n}\leq  \frac{n+1}{2n} - \frac{1}{d+1},
\end{equation}
and  the asymptotical upper bound for $k$ is $\frac{n}{2}$.}

\end{rem}

\begin{rem} \label{augmtm}{\rm  In the case when we add artificial data points to the original data set,
$n$ in (\ref{alg1}) should be replaced with $N$, where $N$ is the cardinality of the augmented data set.
The upper bound for $k$ in (\ref{alg1}) remains the same. }
\end{rem}

The rates of complexity with respect to $n$ and $d$ of Theorem \ref{complexitytm} are confirmed by simulation results presented in Figures
\ref{fig:data_time2} and  \ref{fig:dim_time1}. Measurements are taken on simulated samples of size $n$ from $d$-dimensional
$\mathcal N (0,I)$ distribution, where $I$ is the unit $d\times d$ matrix,   with $d =2,...,10$ and
$ n =$ 40, 80, 160, 320, 640, 1280, 2560, 3000, 3500, 4000, 4500, 5000, 5500, 6000, 6500, 7000.
 The results are averaged on $10$ repetitions for each fixed pair $(d,n)$.

\begin{figure}[!htb]
    \centering

    \includegraphics[width=0.7\textwidth, height=0.4\textheight]{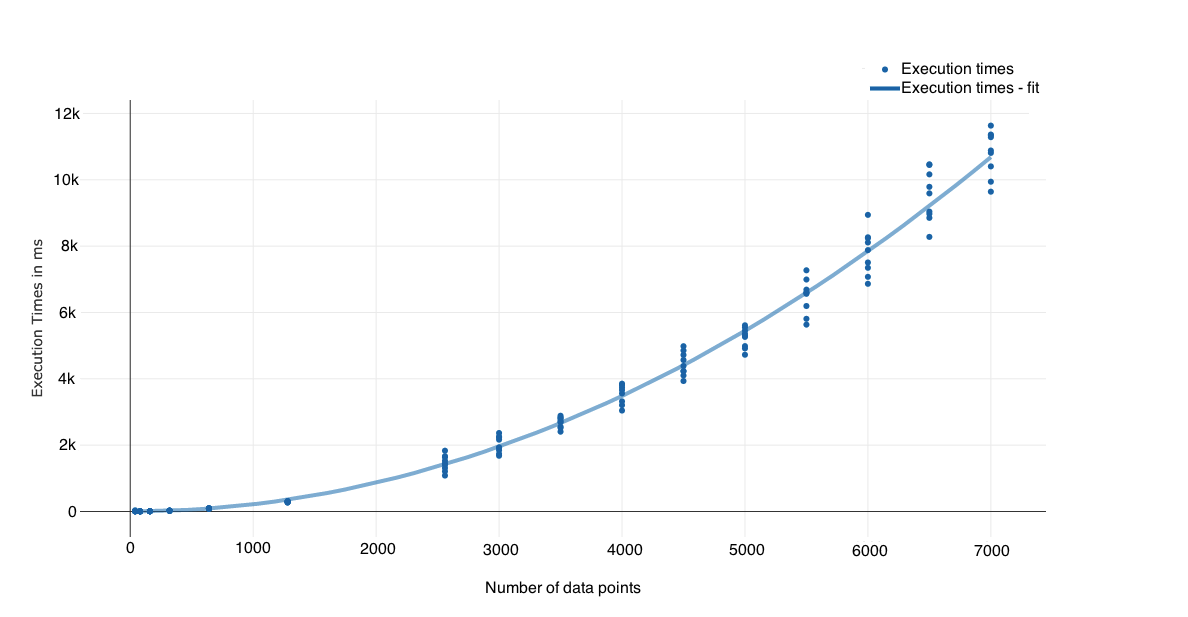}
    \caption{When number of points increases, the execution time grows with the order of $n^2\log n$.}
        \label{fig:data_time2}
\end{figure}

\begin{figure}[!htb]
    \centering
    \includegraphics[width=0.7\textwidth, height=0.4\textheight]{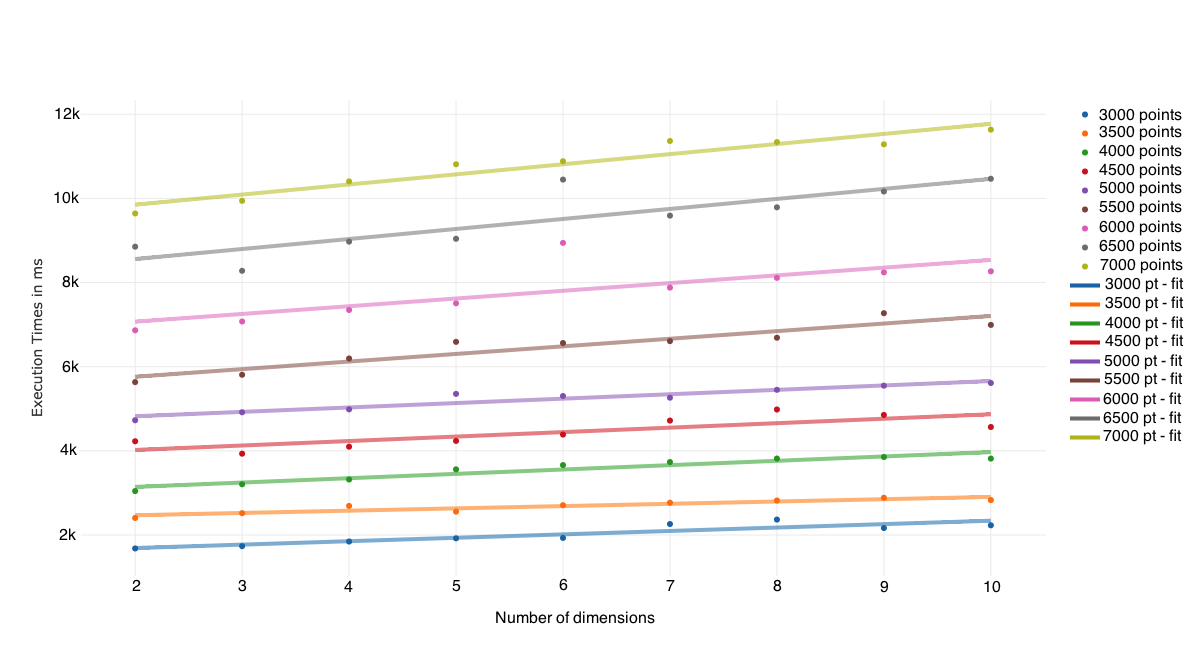}
    \caption{The execution time grows linearly with dimensionality. }
    \label{fig:dim_time1}
\end{figure}

For assessing the accuracy of the median approximation, we also use simulated data sets from $\mathcal N (0,I)$. Let $\hat{m}$ be an approximate median
obtained as the output of the algorithm. Knowing that the median of $\mathcal N (0,I)$
is at origin, and that the squared distance to origin has $\chi^2(d)$ distribution, it is convenient to take $p$-values
$P(\chi^2(d) \leq \left \| \hat{m} \right \| ^2)$ as a measure of error. For samples of size $n=1000$  in ten selected dimensions, the
measurements are performed using ABCDEPTH and DEEPLOC on each sample. The results  summarized in Table \ref{chisquared}, and graphically presented in
figure \ref{p_value}, show the superiority of our algorithm and advantage that increases with dimensionality.

\begin{table}[H]
\begin{flushleft}
\caption{Error size of approximate median points in terms of $p$-values $P(\chi^2(d)\leq \left \| \hat{m} \right \| ^2)$}
\label{chisquared}
\resizebox{0.7\textwidth}{!}{\begin{minipage}{\textwidth}
\begin{tabular}{llllllllllll}
\hline
\multirow{2}{*}{n} & \multirow{2}{*}{Algorithm}                                 & \multicolumn{10}{c}{d}                                                                                                                                                                                                                                                                                                                                                                                                                                                                                                                                                                           \\ \cline{3-12}
                   &                                                            & 4                                                       & 5                                                       & 6                                                       & 7                                                       & 8                                                       & 9                                                       & 10                                                      & 20                                                      & 100                                                     & 200                                                    \\ \hline
1000               & \begin{tabular}[c]{@{}l@{}}DEEPLOC\\ ABCDEPTH\end{tabular} & \begin{tabular}[c]{@{}l@{}}0.0021\\ 0.0011\end{tabular} & \begin{tabular}[c]{@{}l@{}}0.0027\\ 0.0016\end{tabular} & \begin{tabular}[c]{@{}l@{}}0.0039\\ 0.0012\end{tabular} & \begin{tabular}[c]{@{}l@{}}0.0041\\ 0.0011\end{tabular} & \begin{tabular}[c]{@{}l@{}}0.0045\\ 0.0015\end{tabular} & \begin{tabular}[c]{@{}l@{}}0.0051\\ 0.0015\end{tabular} & \begin{tabular}[c]{@{}l@{}}0.0056\\ 0.0012\end{tabular} & \begin{tabular}[c]{@{}l@{}}0.0056\\ 0.0012\end{tabular} & \begin{tabular}[c]{@{}l@{}}0.0612\\ 0.0022\end{tabular} & \begin{tabular}[c]{@{}l@{}}0.1234\\ 0.002\end{tabular} \\ \hline
\end{tabular}
\end{minipage}}
\end{flushleft}
\end{table}

\begin{figure}[H]
  \includegraphics[width=0.7\textwidth]{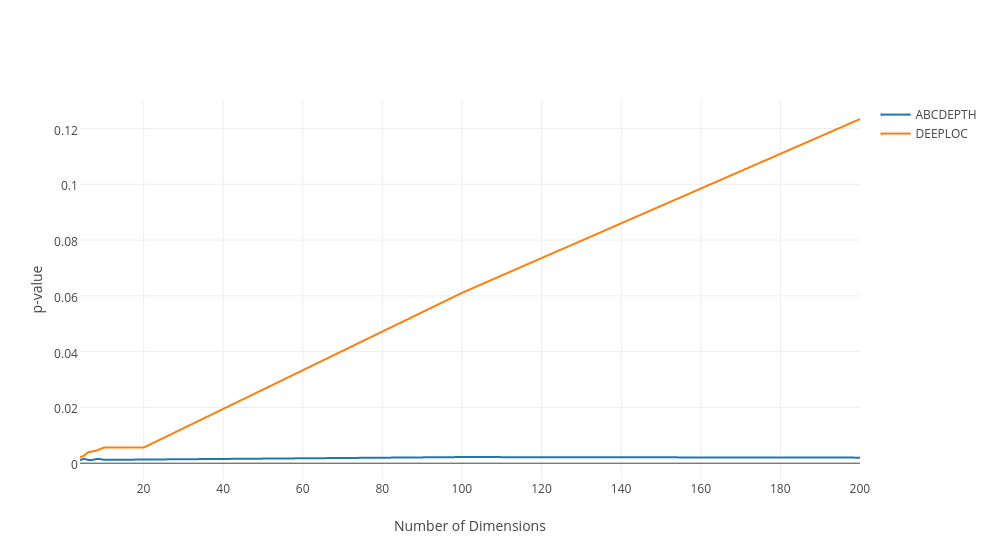}
  \caption{Plotted values from Table \ref{chisquared} }
\label{p_value}
\end{figure}

\subsubsection{Examples}
\label{secexamplestm}


The first simple  example considers $n$ points in dimension $1$ generated from normal $\mathcal N(0,1)$ distribution.
By running ABCDepth in this case with $n=1000$, we got two points (as expected) in the median level set,
$S_{0.5}=\{-0.00314, 0.00034\}$. With another sample with  $n=1001$ (odd number)  from
 the same distribution, the median  set was a singleton, $S_{0.5} = \{0.0043\}$

Now, we demonstrate data sets generated from bivariate and multivariate normal distribution.
Figure \ref{1000_2}  and Figure \ref{1000_3} show the median calculated from $1000$ points in
 dimension $2$ and $3$, respectively from normal $\mathcal N(0,I)$ distribution. Starting from $\alpha=\frac{1}{d+1}$ (see Theorem \ref{alpha})
 the algorithm produces $\sim200$ levels sets for $d=2$, and $\sim300$ level sets for $d=3$, so not all of them are plotted.
 On both figures the median is represented as a black point with depth $\frac{499}{1000}$ on Figure \ref{1000_2}, i.e. $\frac{493}{1000}$
 on Figure \ref{1000_3}.
\begin{figure}[!htb]
    \centering
    \includegraphics[width=0.7\textwidth, height=0.4\textheight]{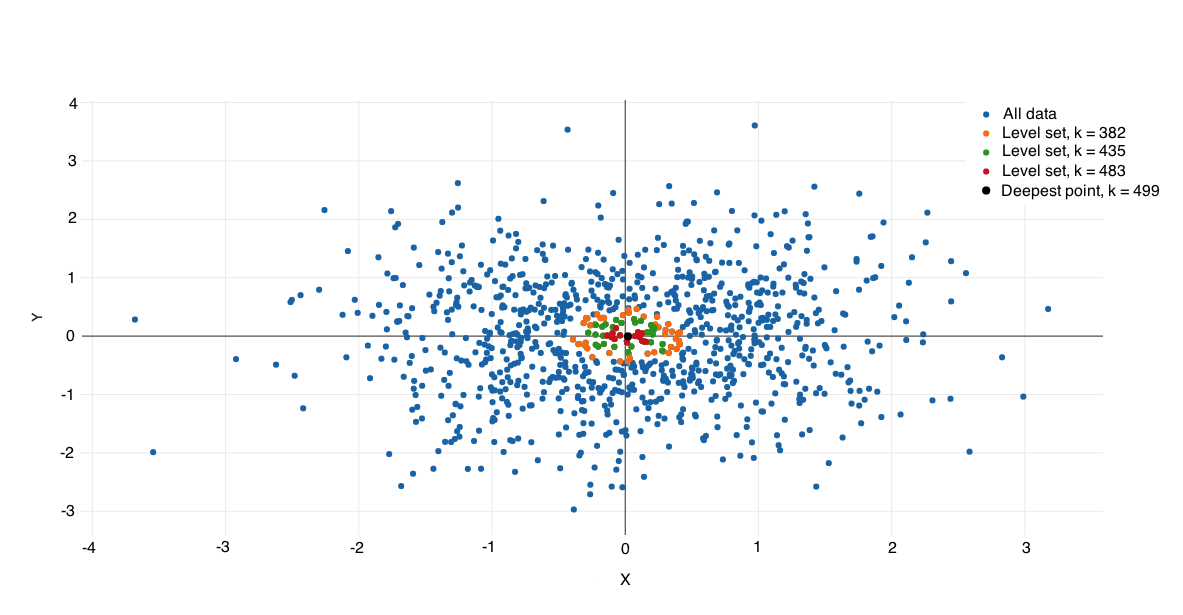}
    \caption{Bivariate normal distribution - four level sets, where the black point at the center is the deepest point.}
    \label{1000_2}
\end{figure}

 \begin{figure}[!htb]
    \centering
    \includegraphics[width=0.7\textwidth]{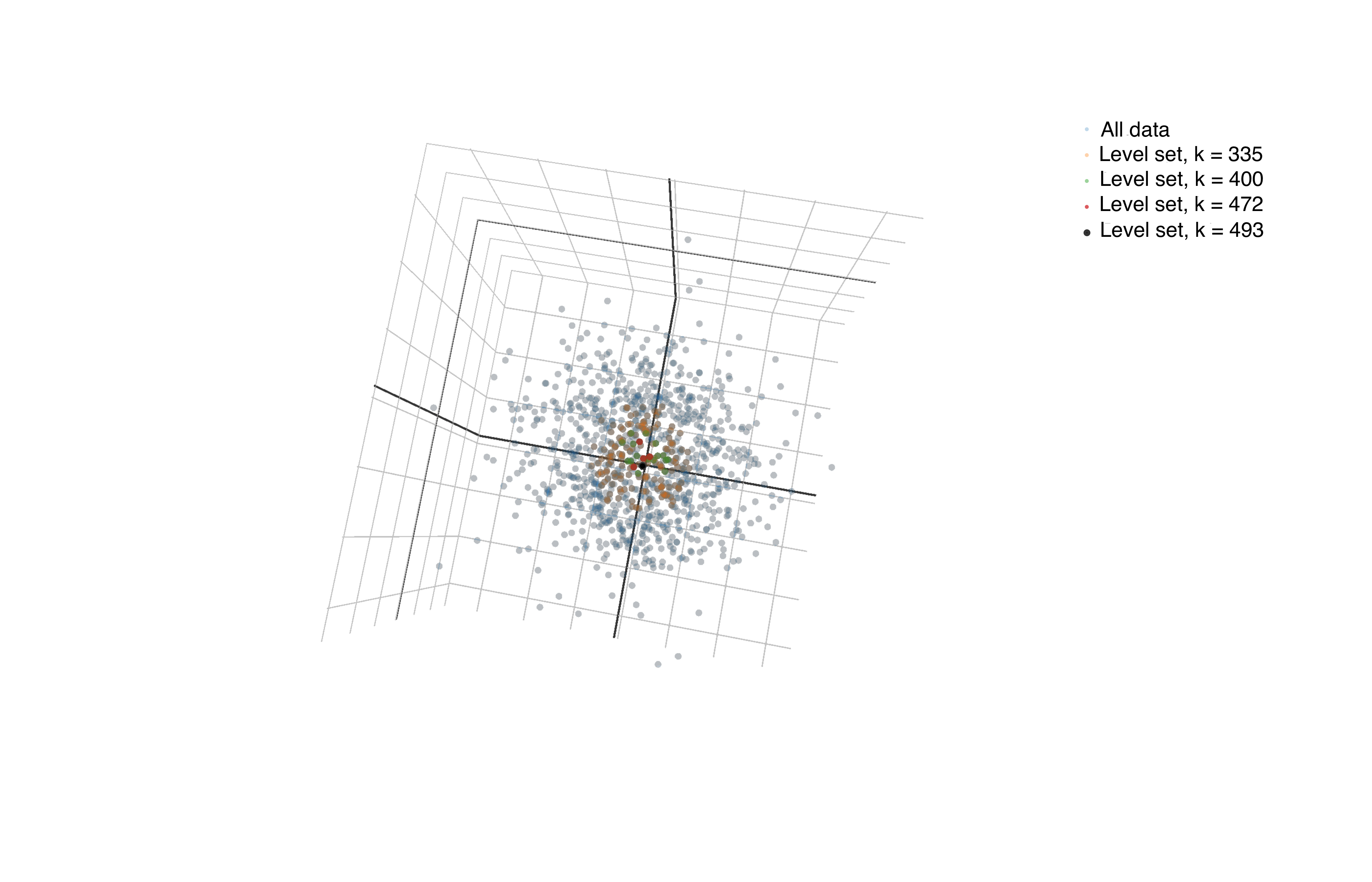}
    \caption{3D normal distribution - four level sets, where the black point at the center is the deepest point.}
    \label{1000_3}
\end{figure}

All data generators that we use in this paper in order to verify and plot the algorithm output were presented at
\cite{bome14}, and they are available within  an open source project at \url{https://bitbucket.org/antomripmuk/generators}.

As a real data example, we take a data set which is rather sparse. The data set is taken from \cite{struro00}, and it has been used in several
other papers as a benchmark.  It contains 23 four-dimensional observations in period from 1966 to 1967 that represent seasonally adjusted
changes in auto thefts in New York city. For the sake of clarity, we take only two dimensions:
percent changes in manpower, and seasonally adjusted changes in auto thefts.
The data is downloaded from \url{http://lib.stat.cmu.edu/DASL/Datafiles/nycrimedat.html}.
Figure \ref{single_ny} shows the output of ABCDepth algorithm if we consider only points from the sample (orange point).
Obviously, the approximate median belongs to the original data set.
Then, we run ABCDepth algorithm  with  $1000$ artificial data points from the uniform distribution as  explained in
Section \ref{sect} and earlier in this section.
The approximate median obtained by this run (green point) has the same depth of $\frac{9}{23}$ as the median
calculated using DEEPLOC algorithm \cite{struro00} by running their Fortran code (red point). We check depths of those two points
(green and red) applying \textit{depth} function based on \cite{struyf98} and implemented in R "depth" package \cite{depthr}.
Evidently, the median, in this case, is not a singleton, i.e. there is more than one point with depth $\frac{9}{23}$.
By adding more than $1000$ artificial points, we can get more than one median point. We will discuss this example again in subsection \ref{exampledepth}.

 \begin{figure}[!htb]
    \centering
    \includegraphics[width=0.7\textwidth]{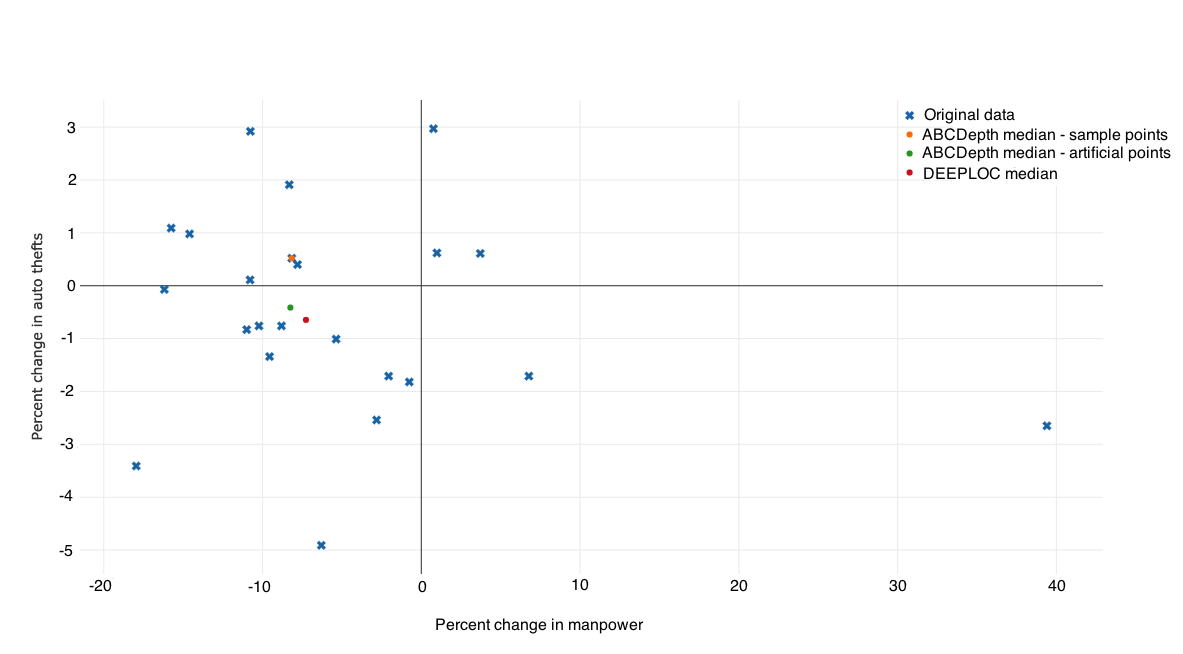}
    \caption{NY crime data set, comparison of Tukey medians using ABCDepth and DEEPLOC.}
    \label{single_ny}
\end{figure}

Another two examples are chosen from \cite{rouru96b}. Figure \ref{artificial_animals} shows $27$ two-dimensional observations
that represent animals brain weight (in g) and the body weight (in kg) taken from \cite{animals}. In order to represent the same
data values, we plotted the logarithms of those measurements as done in \cite{rouru96b}.

Figure \ref{artificial_aircraft} considers the weight and the cost of $23$ single-engine aircraft built between $1947-1979$.
This data set is taken from \cite{aircraft}.

As  in Figure \ref{single_ny}, in those two figures the orange point is the median obtained by running ABCDepth algorithm
using only sample data. Green and red points represent outputs of ABCDepth algorithm applied by adding $1000$ artificial data points from the uniform
 distribution and DEEPLOC median, respectively. These two examples show the importance of out-of-sample points in finding the depth levels and
 Tukey's median.

  \begin{figure}[H]
    \centering
    \includegraphics[width=0.7\textwidth]{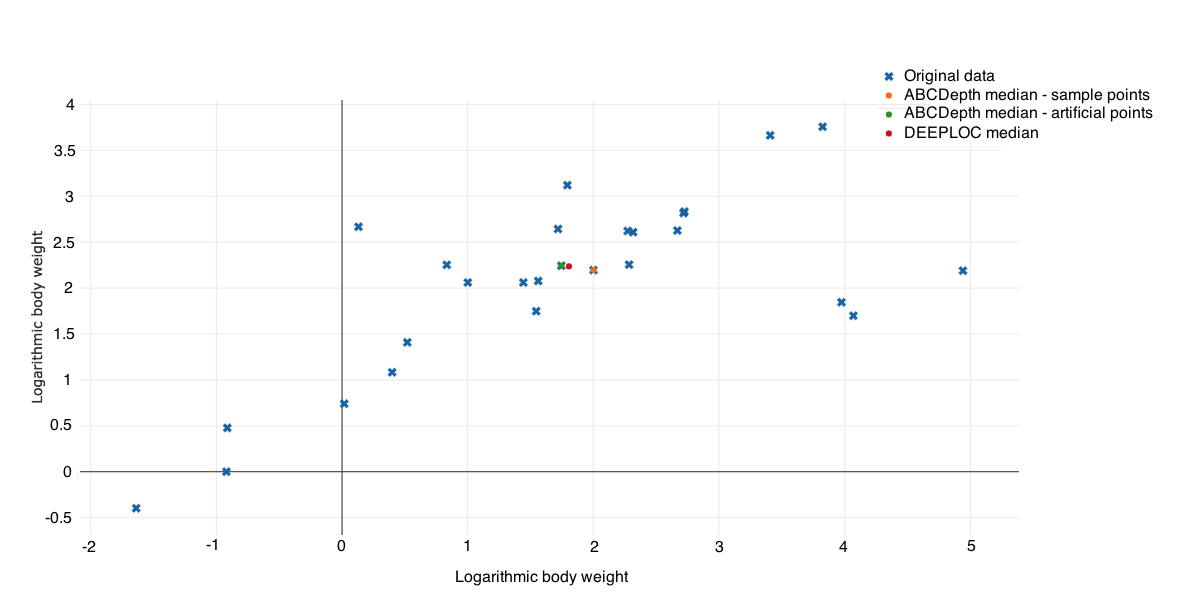}
    \caption{Animals data set, comparison of Tukey medians using ABCDepth and DEEPLOC.}
    \label{artificial_animals}
\end{figure}

 \begin{figure}[H]
    \centering
    \includegraphics[width=0.7\textwidth]{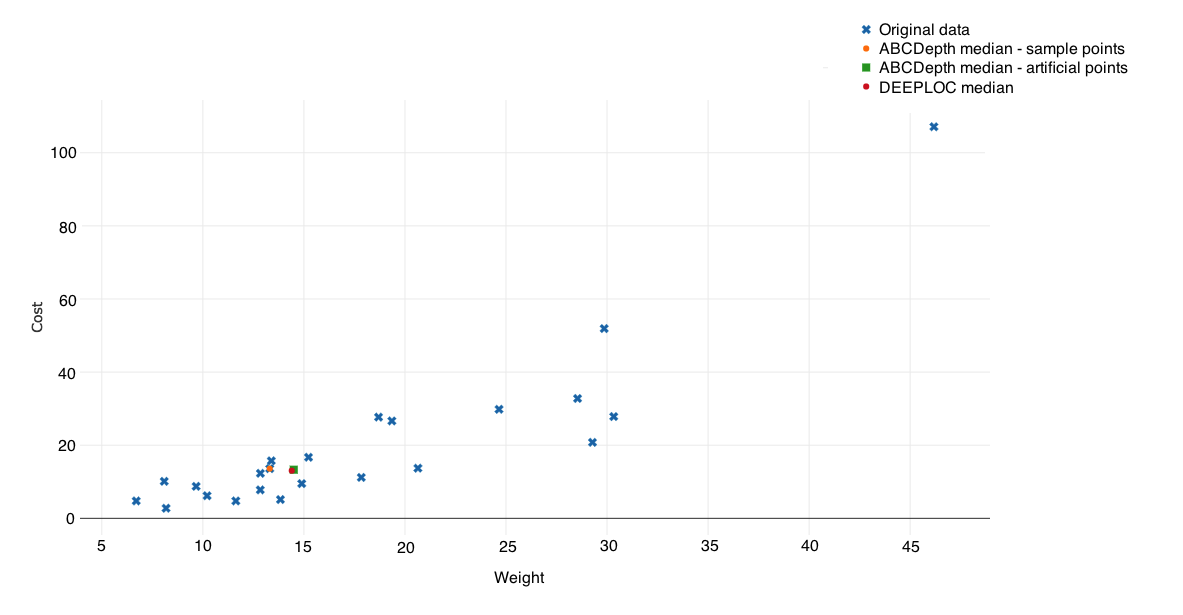}
    \caption{Aircraft data set, comparison of Tukey medians using ABCDepth and DEEPLOC.}
    \label{artificial_aircraft}
\end{figure}

\subsection{Adapted Implementation: Algorithm 2 for finding the Tukey's depth of a sample point and out-of-sample point}
\label{subsectd}

Let us recall that by  Corollary \ref{deviale}, a point $\boldsymbol x$ has depth $h$ if and only if $\boldsymbol x \in S_{\alpha}$ for $\alpha \leq h$
and $\boldsymbol x \not \in S_{\alpha} $ for $\alpha > h$. With a sample of size $n$, we can consider only $\alpha =\frac{k}{n}$, $k=1,\ldots, n$, because
for $\frac{k-1}{n} <\alpha< \frac{k}{n}$, we have that $D(\boldsymbol x)\geq \alpha \iff D(\boldsymbol x)\geq \frac{k}{n}$. Therefore, the statement of Corollary
\ref{deviale} adapted to the sample distribution can be formulated as (using the fact that $S_{\beta} \subset S_{\alpha}$ for $\alpha <\beta$):
\begin{equation}
\label{deviales}
D(\boldsymbol x)= \frac{k}{n} \iff \boldsymbol x \in S_{\frac{k}{n}} \quad {\rm and} \quad \boldsymbol x \not\in S_{\frac{k+1}{n}} .
\end{equation}

From (\ref{deviales}) we derive the algorithm for Tukey's depth of a sample point $\boldsymbol x$ as follows. Let $\alpha_k = \frac{k}{n}$. The level set
 $S_{\alpha_1}$
contains all points in the sample. Then we construct $S_{\alpha_2}$ as an intersection of $n$ balls that contain $n-1$ sample points.
If $\boldsymbol x \not \in S_{\alpha_2}$, we conclude that $D(\boldsymbol x)=1/n$, and stop. Otherwise, we iterate this procedure till we get the situation as in
right side of (\ref{deviales}), when we conclude that the depth is $\frac{k}{n}$. The output of the algorithm is $k$.

\begin{rem}{\rm As in Remark \ref{kmaxgp}, it can be shown that the approximate depth $k/n$ is never greater than the true depth.

}
\end{rem}

Implementation-wise, in order to improve the algorithm complexity, we do not need to construct the level sets.
It is enough to count balls that contain point $\boldsymbol x$. The algorithm stops when for some $k$, there
exists at least one ball (among the candidates for the intersection)  that does not contain $\boldsymbol x$.
Thus, the depth of the point $\boldsymbol x$ is $k-1$.

With a very small modification, the same algorithm can be applied to a point $\boldsymbol x$ out of the sample.
We can just treat $\boldsymbol x$ as an artificial point,
in the same way as in previous sections. That is, the size of the required balls has to be $n-k+1$ points from the sample,
not counting $\boldsymbol x$. The rest of
the algorithm is the same as in the case of a sample point $\boldsymbol x$.

In both versions (sample or out-of-sample), we can use additional artificial points to increase the precision. The sample version of the algorithm is
detailed below.

\begin{algorithm}[H]
  \LinesNumbered
  \SetAlgoLined
  \KwData{Original data, $X_n = (\boldsymbol{x_1}, \boldsymbol{x_1},...,\boldsymbol{x_n}) \in \mathbb{R}^{d \times n}$, $\boldsymbol x=\boldsymbol x_i$ for a fixed $i$ - the data point whose depth is calculated.}
  \KwResult{Tukey depth at $\boldsymbol x$.}
  \BlankLine
   \tcc{Iteration Phase}
   $S_{\alpha_1} = \{\boldsymbol x_1,\ldots,\boldsymbol x_n\}$\;
   \For {$k \gets 2$ \textbf{to} $n$} {

  	$p=0$ - Number of balls that contain $\boldsymbol x$. Its initial value is $0$ \;
    	\tcc{Find balls that contain point $\boldsymbol x$}
   	\For {$i \gets 1$ \textbf{to} $n$} {
   		\If {$\boldsymbol x \in B_i$, where $B_i$ contains $n - k + 1$ original data points} {
			$p=p+1$;
		}
   	}
	\If {$p \neq n$} {
		\Return $k-1$	
	}

   }
  \caption{Calculating Tukey depth of a sample point.}
\end{algorithm}

\subsubsection{Complexity}
\begin{thm}
\label{complexitytd}
Adapted ABCDepth algorithm for finding approximate Tukey depth of a sample point has order of $O(dn^2 + n^2\log{n})$ time complexity.
\end{thm}
\begin{proof}
Balls construction for Algorithm 2 is the same as in Algorithm 1 (lines 1-10), so by Theorem \ref{complexitytm}
this part runs in  $O(dn^2 + n^2 \log n)$ time.
For the point with the depth $\alpha_k$, algorithm enters in iteration loop $k$ times and it iterates through all $n$ points to
find the balls that contain point $\boldsymbol x$, so the whole iteration phase runs in $O(kn)$ time.

Overall complexity of the  Algorithm 2 is:
\begin{equation}
\label{alg2}
O(dn^2 + n^2 \log n) + O(kn) \sim O(dn^2 + n^2 \log n),
\end{equation}
which had to be proved. \end{proof}

\begin{rem}  { \rm When the input data set is sparse or when the sample set is small, we add artificial data to the original data set in order to
improve the algorithm accuracy. In that case, $n$ in (\ref{alg2}) should be replaced with $N$. }
\end{rem}

\subsubsection{Examples}
\label{exampledepth}

To illustrate the output for the Algorithm 2, we use the same real data sets as we used in Figures \ref{single_ny}-\ref{artificial_aircraft}.
For all data sets we applied Algorithm 2 in two runs; first time with sample points only, and second time with additional
 $1000$ artificial points generated from uniform distribution. Points' depths are verified using \textit{depth} function from \cite{struyf98}
 implemented in \cite{depthr}. For each data set, we calculate the accuracy as $\frac{100k}{n}\%$, where $n$ is the sample size and $k$ is the
 number of points that has the
 correct depth compared with algorithm presented in \cite{struyf98}.

In Figure \ref{ny_depths_sample} we showed NY crime points depths with accuracy of $26\%$, but if we add more points to
the original data set as we showed on Figure \ref{ny_depths_artificial}, the accuracy is greatly improved, to $87\%$.

\begin{figure}[H]
\centering
\centering
\includegraphics[width=0.7\textwidth]{nycrime_depths_1.pdf}
\caption{NY crime data - point depths using only original data.}
\label{ny_depths_sample}
\end{figure}
 \vspace{-4.5cm}
\begin{figure}[H]
\centering
\includegraphics[width=0.7\textwidth]{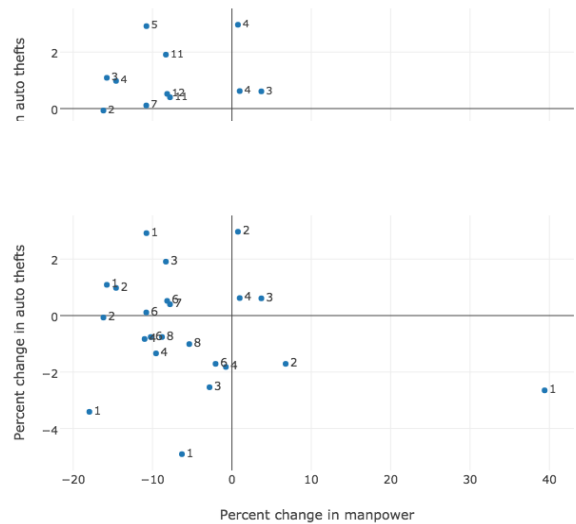}
\caption{NY crime data - point depths using original  and artificial data.}
\label{ny_depths_artificial}
\end{figure}

Figure \ref{animals_depth_sample} shows the same accuracy
of $26\%$ for animals data set, in the case when Algorithm 2 is run with sample points only.
 By adding more points as in Figure \ref{animals_depth_artificial}, the accuracy is improved to $100\%$.

\begin{figure}[H]
\centering
\centering
\includegraphics[width=0.7\textwidth]{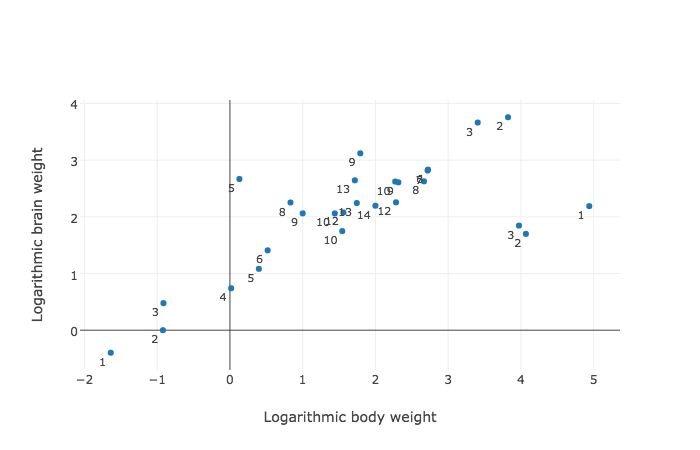}
\caption{Animals data - point depths using only original data.}
\label{animals_depth_sample}
\end{figure}
\begin{figure}[H]
\centering
\includegraphics[width=0.7\textwidth]{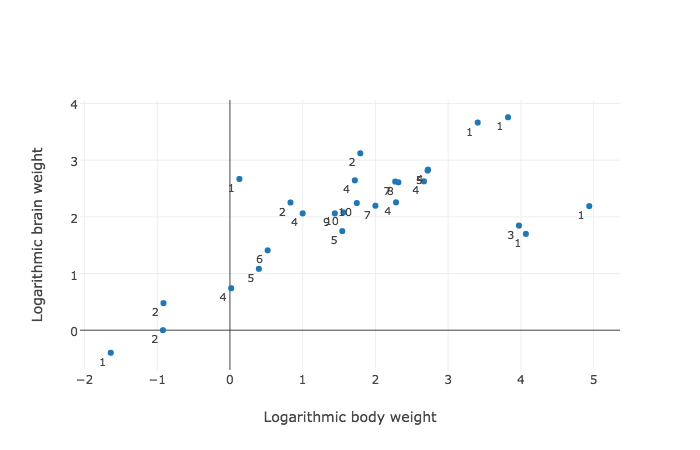}
\caption{Animals data -  point depths using  original and artificial data.}
\label{animals_depth_artificial}
\end{figure}

The third example is aircraft data set presented in Figure \ref{aircraft_depth_sample} and Figure \ref{aircraft_depth_artificial}.
The accuracy with artificial points is $95\%$, otherwise it is  $18\%$.

\begin{figure}[H]
\centering
\centering
\includegraphics[width=0.7\textwidth]{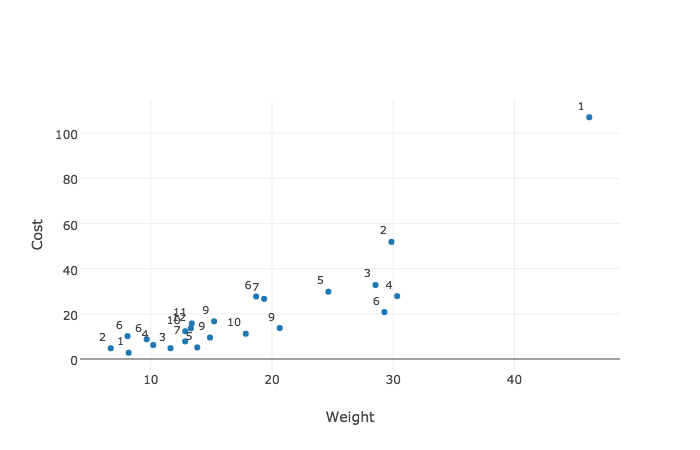}
\caption{Aircraft data - point depths using only original data.}
\label{aircraft_depth_sample}
\end{figure}
\begin{figure}[H]
\centering
\includegraphics[width=0.7\textwidth]{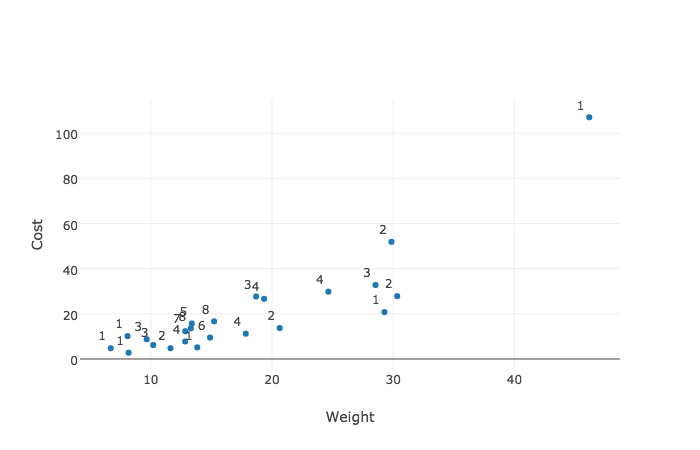}
\caption{Aircraft data - point depths using original  and artificial data.}
\label{aircraft_depth_artificial}
\end{figure}

As the last example of this section, we would like to calculate depths of the points plotted on Figure
\ref {single_ny} using ABCDepth Algorithm 2.  In Figure \ref {single_ny} we plotted Tukey median for NY crime data
set using Algorithm 1 with artificial data points (green point) and compared the result with the median obtained by
DEEPLOC (red point). Both points are out of the sample. In Figure \ref{ny_medians_depths}, we show depths of all sample points
including the depths of two median points, all attained by ABCDepth Algorithm 2.
Algorithm presented in \cite{struyf98} and ABCDepth Algorithm 2 calculate the same depth value for both median points.

\begin{figure} [H]
\centering
\includegraphics[width=0.7\textwidth]{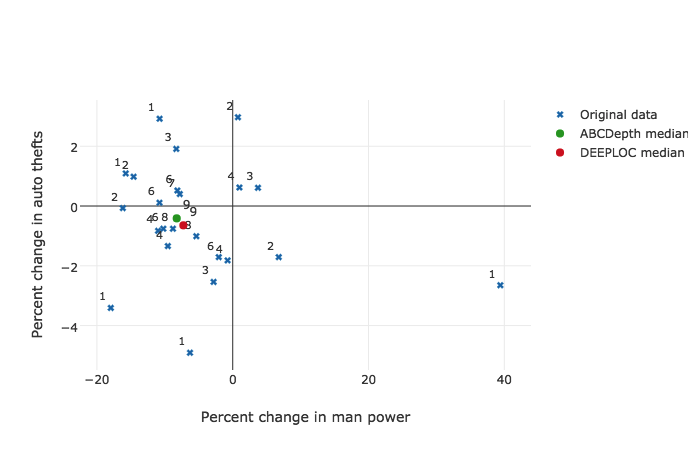}
\caption{NY crime data - point depths using  original  and artificial data.}
\label{ny_medians_depths}
\end{figure}


\section{Performance and Comparisons}
\label{secc}

According to Theorem \ref{complexitytm}, the complexity of calculating Tukey median grows linearly with the dimension and in terms
of a number of data points,
it grows with the order of $n^2\log n$.
Rousseeuw and Ruts in \cite{rouru98} pioneered with an exact algorithm called HALFMED for Tukey median in two dimensions
that runs in $O(n^2 log^2{n})$ time. This algorithm is better than ABCDepth for $d=2$, but it processes only bivariate data sets.
Struyf and Rousseeuw in \cite{struro00} implemented the first approximate algorithm called DEEPLOC for finding the deepest location
in higher dimensions. Its complexity is $O(kmn \log{n} +kdn+md^3+mdn)$ time, where $k$ is the number of steps taken by the program and
$m$ is the number of directions, i.e. vectors constructed by the program. This algorithm is very efficient for low-dimensional data sets,
but for high-dimensional data sets ABCDepth algorithm outperforms DEEPLOC.
Chan in \cite{chan04} presents an approximate randomized algorithm for maximum Tukey depth. It runs in $O(n^{d-1})$ time but it has
not been implemented yet.

In Table 2 execution times of DEEPLOC algorithm and ABCDepth algorithm for finding Tukey median are reported.
The measurements are performed using synthetic data generated from the multivariate normal $\mathcal N(0,1)$ distribution.
In this table, we demonstrate how ABCDepth algorithm behaves with thousands of high-dimensional data points. It takes
$\sim 13$ minutes for $n=7000$ and $d=2000$. Since DEEPLOC algorithm does not support data sets with $d > n$ and returns the error message:
"the dimension should be at most the number of objects", we denoted those examples with $-$ sign in the table.
The sign $*$ means that the median is not computable at least once in $12$ hours.

\begin{table}[H]
\begin{flushleft}
\caption{Comparison between DEEPLOC and ABCDepth execution times in seconds.}
\label{deepABC}
\resizebox{0.58\textwidth}{!}{\begin{minipage}{\textwidth}
\begin{tabular}{@{}lllllllllllllll@{}}
\toprule
\multirow{2}{*}{d} & \multirow{2}{*}{Algorithm}                                 & \multicolumn{13}{c}{n}                                                                                                                                                                                                                                                                                                                                                                                                                                                                                                                                                                                                                                                                                                                                                    \\ \cmidrule(l){3-15}
                   &                                                            & 320                                                  & 640                                                     & 1280                                                 & 2560                                                 & 3000                                                 & 3500                                                   & 4000                                                  & 4500                                                   & 5000                                                  & 5500                                                   & 6000                                                   & 6500                                                   & 7000                                                                   \\ \midrule
50                 & \begin{tabular}[c]{@{}l@{}}Deeploc\\ ABCDepth\end{tabular} & \begin{tabular}[c]{@{}l@{}}4.43\\ 0.15\end{tabular}  & \begin{tabular}[c]{@{}l@{}}7.15\\ 0.63\end{tabular}     & \begin{tabular}[c]{@{}l@{}}12.65\\ 2.86\end{tabular} & \begin{tabular}[c]{@{}l@{}}23.87\\ 4.95\end{tabular} & \begin{tabular}[c]{@{}l@{}}30.93\\ 7.27\end{tabular} & \begin{tabular}[c]{@{}l@{}}31.79\\ 8.65\end{tabular}   & \begin{tabular}[c]{@{}l@{}}37.66\\ 12.51\end{tabular} & \begin{tabular}[c]{@{}l@{}}45.35\\ 14.18\end{tabular}  & \begin{tabular}[c]{@{}l@{}}50.72\\ 17.51\end{tabular} & \begin{tabular}[c]{@{}l@{}}63.13\\ 22.18\end{tabular}  & \begin{tabular}[c]{@{}l@{}}63.75\\ 25.86\end{tabular}  & \begin{tabular}[c]{@{}l@{}}84.13\\ 29.24\end{tabular}  & \begin{tabular}[c]{@{}l@{}}69.61\\ 37.34\end{tabular}                  \\ \midrule
100                & \begin{tabular}[c]{@{}l@{}}Deeploc\\ ABCDepth\end{tabular} & \begin{tabular}[c]{@{}l@{}}19.42\\ 0.22\end{tabular} & \begin{tabular}[c]{@{}l@{}}22.85\\ 0.92\end{tabular}    & \begin{tabular}[c]{@{}l@{}}33.81\\ 2.03\end{tabular} & \begin{tabular}[c]{@{}l@{}}77.45\\ 7.83\end{tabular} & \begin{tabular}[c]{@{}l@{}}69.04\\ 9.78\end{tabular} & \begin{tabular}[c]{@{}l@{}}105.56\\ 13.14\end{tabular} & \begin{tabular}[c]{@{}l@{}}97.39\\ 17.89\end{tabular} & \begin{tabular}[c]{@{}l@{}}120.05\\ 23.52\end{tabular} & \begin{tabular}[c]{@{}l@{}}140.04\\ 30.6\end{tabular} & \begin{tabular}[c]{@{}l@{}}131.85\\ 39.18\end{tabular} & \begin{tabular}[c]{@{}l@{}}127.36\\ 49.03\end{tabular} & \begin{tabular}[c]{@{}l@{}}212.42\\ 68.46\end{tabular} & \begin{tabular}[c]{@{}l@{}}183.27\\ 82.02\end{tabular}                 \\ \midrule
500                & \begin{tabular}[c]{@{}l@{}}Deeploc\\ ABCDepth\end{tabular} & \begin{tabular}[c]{@{}l@{}}-\\ 0.693\end{tabular}    & \begin{tabular}[c]{@{}l@{}}1616.53\\ 3.181\end{tabular} & \begin{tabular}[c]{@{}l@{}}*\\ 8.4\end{tabular}      & \begin{tabular}[c]{@{}l@{}}*\\ 27.9\end{tabular}     & \begin{tabular}[c]{@{}l@{}}*\\ 41.61\end{tabular}    & \begin{tabular}[c]{@{}l@{}}*\\ 53.73\end{tabular}      & \begin{tabular}[c]{@{}l@{}}*\\ 71.95\end{tabular}     & \begin{tabular}[c]{@{}l@{}}*\\ 89.36\end{tabular}      & \begin{tabular}[c]{@{}l@{}}*\\ 109.22\end{tabular}    & \begin{tabular}[c]{@{}l@{}}*\\ 140.18\end{tabular}     & \begin{tabular}[c]{@{}l@{}}*\\ 151.45\end{tabular}     & \begin{tabular}[c]{@{}l@{}}*\\ 180.5\end{tabular}      & \begin{tabular}[c]{@{}l@{}}*\\ 213.01\end{tabular}                     \\ \midrule
1000               & \begin{tabular}[c]{@{}l@{}}Deeploc\\ ABCDepth\end{tabular} & \begin{tabular}[c]{@{}l@{}}-\\ 1.165\end{tabular}    & \begin{tabular}[c]{@{}l@{}}-\\ 3.99\end{tabular}        & \begin{tabular}[c]{@{}l@{}}*\\ 14.389\end{tabular}   & \begin{tabular}[c]{@{}l@{}}*\\ 54.18\end{tabular}    & \begin{tabular}[c]{@{}l@{}}*\\ 74.38\end{tabular}    & \begin{tabular}[c]{@{}l@{}}*\\ 98.73\end{tabular}      & \begin{tabular}[c]{@{}l@{}}*\\ 129.85\end{tabular}    & \begin{tabular}[c]{@{}l@{}}*\\ 164.96\end{tabular}     & \begin{tabular}[c]{@{}l@{}}*\\ 203.37\end{tabular}    & \begin{tabular}[c]{@{}l@{}}*\\ 246.54\end{tabular}     & \begin{tabular}[c]{@{}l@{}}*\\ 286.17\end{tabular}     & \begin{tabular}[c]{@{}l@{}}*\\ 344.94\end{tabular}     & \begin{tabular}[c]{@{}l@{}}*\\ 39.16\end{tabular}                      \\ \midrule
2000               & \begin{tabular}[c]{@{}l@{}}Deeploc\\ ABCDepth\end{tabular} & \begin{tabular}[c]{@{}l@{}}-\\ 2.21\end{tabular}     & \begin{tabular}[c]{@{}l@{}}-\\ 7.86\end{tabular}        & \begin{tabular}[c]{@{}l@{}}-\\ 27.25\end{tabular}    & \begin{tabular}[c]{@{}l@{}}*\\ 107.46\end{tabular}   & \begin{tabular}[c]{@{}l@{}}*\\ 132.77\end{tabular}   & \begin{tabular}[c]{@{}l@{}}*\\ 180.02\end{tabular}     & \begin{tabular}[c]{@{}l@{}}*\\ 243.1\end{tabular}     & \begin{tabular}[c]{@{}l@{}}*\\ 297.6\end{tabular}      & \begin{tabular}[c]{@{}l@{}}*\\ 386.75\end{tabular}    & \begin{tabular}[c]{@{}l@{}}*\\ 475.87\end{tabular}     & \begin{tabular}[c]{@{}l@{}}*\\ 554.23\end{tabular}     & \begin{tabular}[c]{@{}l@{}}*\\ 666.4\end{tabular}      & \multicolumn{1}{c}{\begin{tabular}[c]{@{}c@{}}*\\ 764.74\end{tabular}} \\ \bottomrule
\end{tabular}
\end{minipage}}
\end{flushleft}
\end{table}

ABCDepth algorithm for finding Tukey depth of a point runs in $O(dn^2 + n^2 \log n)$ as we showed in Theorem \ref{complexitytd}.
Most of the algorithms for finding Tukey depth are exact and at the same time computationally expensive. One of the first exact
algorithms for bivariate data sets, called LDEPTH, is proposed by Rousseeuw and Ruts in \cite{rouru96}. It has complexity of
$O(n \log n)$ and like HALFMED, it outperforms ABCDepth for $d=2$.
Rousseeuw and Struyf in \cite{struyf98} implemented an exact algorithm for $d=3$ that runs in $O(n^2 \log n)$ time,
 and an approximate algorithm for $d > 3$ that runs in $O(md^3 + mdn)$, where $m$ is the number of
  directions perpendicular to hyperplanes through $d$ data points.

The later work of Chen et al. in \cite{chmowa13}, presented approximate algorithms, based on the third approximation method of Rousseeuw and Struyf,
 in \cite{struyf98}, reducing the problem from $d$ to $k$ dimensions.

 The first one, for $k=1$, runs in $O(\epsilon^{1-d}dn)$ time and the second one, for $k \geq 2$, runs in $O((\epsilon^{-1} c \log n)^d)$,
 where $\epsilon$ and $c$ are empirically chosen constants.
Another exact algorithm for finding Tukey depth in $\mathbb{R}^d$ is proposed by Liu and Zuo in \cite{zuo14}, which proves
to be extremely time-consuming (see Table 5.1 of Section 5.3 in \cite{pavlo14}) and the algorithm involves heavy computations,
but can serve as a benchmark.
Recently, Dyckerhoff and Mozharovskyi in \cite{pavlo16} proposed two exact algorithms for finding halfspace depth that run in
$O(n^d)$ and $O(n^{d-1} \log n)$ time.

Table 3 shows execution times of ABCDepth algorithm for finding a depth of a sample point. Measurements are
derived from synthetics data from the multivariate standard normal distribution. Execution time for each data set represents
averaged time consumed per data point. Most of the execution time ($\sim 95\%$) is spent on balls construction (see lines 1-10 of the Algorithm 1),
while finding a point
depth itself (iteration phase of the Algorithm 2) is really fast since it runs in $O(kn)$ time.

\begin{table}[]
\begin{flushleft}
\caption{Average time per data point.}
\label{depthABC}
\resizebox{0.7\textwidth}{!}{\begin{minipage}{\textwidth}
\begin{tabular}{@{}llllllllllllll@{}}
\toprule
\multirow{2}{*}{d} & \multicolumn{13}{c}{n}                                                                                                      \\ \cmidrule(l){2-14}
                   & 320  & 640  & 1280 & 2560  & 3000  & 3500  & 4000  & 4500  & 5000   & 5500   & 6000   & 6500   & 7000                       \\ \midrule
50                 & 0.07 & 0.21 & 1.21 & 8.23  & 12.64 & 19.22 & 28.56 & 42.04 & 64.33  & 77.45  & 98.79  & 121.86 & 150.73                     \\ \midrule
100                & 0.08 & 0.25 & 1.23 & 8.18  & 13.91 & 20.48 & 28.51 & 44.31 & 65.55  & 81.91  & 99.96  & 123.84 & 154.65                     \\ \midrule
500                & 0.13 & 0.42 & 1.84 & 11.42 & 17.93 & 21.42 & 35.41 & 52.07 & 73.21  & 95.18  & 119.82 & 141.88 & 176.12                     \\ \midrule
1000               & 0.17 & 0.53 & 2.52 & 13.53 & 20.13 & 32.35 & 41.71 & 58.72 & 82.92  & 103.84 & 138.69 & 155.32 & 200.55                     \\ \midrule
2000               & 0.26 & 0.94 & 4.12 & 18.32 & 28.12 & 38.79 & 56.04 & 73.79 & 102.98 & 124.48 & 156.54 & 186.59 & \multicolumn{1}{c}{232.45} \\ \bottomrule
\end{tabular}\end{minipage}}
\end{flushleft}
\end{table}

The ABCDepth algorithms have been implemented in Java. Tests for all algorithms are run using one kernel of Intel Core i7 (2.2 GHz) processor.
Computational codes are available from the authors upon request.

\section{Concluding Remarks}

There is no doubt that exact algorithms are needed,  whether it is about calculating the depth of a point, or a multivariate median.
Those algorithms are precise and serve as an benchmark measurement for all approximate algorithms.
Nowadays, the real life applications such as clustering, classification, outlier detection or, in general,
any kind of data processing, contain at least thousands of multidimensional observations.
In those cases, available exact algorithms are not the best choice - the complexity of the exact algorithms  grows exponentially with
dimension due to projections of sample points to a large number of  directions.
Hence, the exact algorithms are time consuming and often restricted by number of observations and its dimensionality.
 Therefore, for large data sets, approximate algorithms correspond to a good solution.

In this paper we presented approximate ABCDepth algorithms based on a novel \textit{balls intersection} idea explained in
Sections \ref{sect} and \ref{seca} that brings a lots of advantages.
With a small modification, the main idea from \cite{merkle10} is used for implementing two algorithms:
one is for calculating Tukey median and the another one is for calculating Tukey depth of a sample and out-of-sample point.
Using synthetic and real data sets and comparing our performances with those of previous approximate algorithms,
we showed that our algorithms fulfill the following:

\begin{enumerate*}[label=\roman*),itemjoin={{, }},afterlabel=\unskip{{~}}]
\item  high accuracy (see examples in Sections \ref{secexamplestm} and \ref{exampledepth})
\item they are much faster especially for large $n$ and $d$ (see Tables \ref{deepABC} and \ref{depthABC}
as well as Table 5.1 of Section 5.3 in \cite{pavlo14})
\item they can handle a larger number of multidimensional observations; those algorithms are the only algorithms  tested
with data sets that contain up to $n=7000$ and $d=2000$,
\item Algorithm 1 computes multidimensional median with the complexity of $O((d + k)n^2 + n^2\log{n})$
\item Algorithm 2 computes the depth of a single point with complexity of $O(dn^2 + n^2\log{n})$
\item both complexities have linear growth in $d$ and quadratic growth in $n$ (see Figures
\ref{fig:data_time2} and  \ref{fig:dim_time1})
\item an additional theoretical advantage of ABCDepth approach is that the data points are not assumed to be in "general position".
\end{enumerate*}

\medskip

{\bf Disclaimer.} A previous version   of this work was presented in a poster session of CMStatistics2016 and the abstract is
posted and available  in \cite{cmstat}. Otherwise, this paper has not been published in conference proceedings or elsewhere.

\medskip

{\bf Acknowledgements. } We would like to express our gratitude to Anja Struyf and coauthors for sharing the code and the data
that were used in their papers of immense importance in the area. Answering to Yijun Zuo's doubts about the first arXiv version \cite{bm161} of this paper
and solving difficult queries that he was proposing, helped us to improve the
presentation and the algorithms. The second author acknowledges the  support by grants III 44006 and 174024
from Ministry of Education, Science and Technological Development of
Republic of Serbia.

\vspace{.5cm}


\bibliographystyle{acm}

\bibliography{btukey}

\end{document}